\IfFileExists{IEEEtran.cls}{
  \documentclass[conference]{IEEEtran}
}{
  \documentclass{paperreview}
}
\IEEEoverridecommandlockouts
\usepackage{cite}
\usepackage{graphicx}
\usepackage{textcomp}
\usepackage{xcolor}
\usepackage{balance}
\usepackage[cmex10]{amsmath}
\usepackage{amssymb}
\usepackage{amsthm}
\usepackage{amsfonts}
\usepackage{bm}
\usepackage{xfrac}
\usepackage{empheq}
\usepackage{mathtools}
\newtheorem{theorem}{Theorem}

\newtheorem{example}{Example}
\newtheorem{proposition}{Proposition}
\newtheorem{lemma}{Lemma}

\newtheorem{corollary}{Corollary}
\newtheorem{remark}{Remark}
\theoremstyle{definition}

\def\BibTeX{{\rm B\kern-.05em{\sc i\kern-.025em b}\kern-.08em
    T\kern-.1667em\lower.7ex\hbox{E}\kern-.125emX}}
\begin{document}

\title{On Privacy Mechanism Design Under a Bounded $\alpha$-Lift Leakage Constraint with $\chi^2$-Utility
%{\footnotesize \textsuperscript{*}Note: Sub-titles are not captured for https://ieeexplore.ieee.org  and
}

\author{\IEEEauthorblockN{1\textsuperscript{st} Amirreza Zamani}
\IEEEauthorblockA{\textit{Department of Information Science and Engineering, KTH} \\
%\textit{and Engineering}\\
Stockholm, Sweden \\
amizam@kth.se}
\and
\IEEEauthorblockN{2\textsuperscript{nd} Mohammad A. Zarrabian}
\IEEEauthorblockA{\textit{Telstra} \\
%\textit{ }\\
Sydney, Australia \\
amin.zarrabian@team.telstra.com}\and
\IEEEauthorblockN{3\textsuperscript{rd} Parastoo Sadeghi}
\IEEEauthorblockA{\textit{School of Engineering and Technology, UNSW} \\
%\textit{ }\\
Canberra, Australia \\
p.sadeghi@unsw.edu.au}
\and
\IEEEauthorblockN{4\textsuperscript{th} Mikael Skoglund}
\IEEEauthorblockA{\textit{Department of Information Science and Engineering, KTH} \\
%\textit{}\\
Stockholm, Sweden \\
skoglund@kth.se}}
\maketitle
\suppressfloats[t]
\ifdefined\paperreviewmode\thispagestyle{review}\fi
\begin{abstract}
We study the design of a statistical privacy mechanism under a bounded $\alpha$-lift leakage constraint. An agent observes useful data $Y$, which is arbitrarily correlated with private data $X$, and releases $U$. The goal is to maximize Pearson $\chi^2$-information $\chi^2(Y;U)$ while satisfying the point-wise privacy constraint. This utility is an $f$-information that is nonnegative, obeys data processing, and measures the aggregate normalized change from the prior $P_Y$ to the posteriors $P_{Y|U=u}$.
The main benefit of this choice is that $\chi^2(Y;U)$ is exactly quadratic in the posterior perturbations; hence, no local approximation of the utility is required. For the nonlinear $\alpha$-lift constraint, we derive a uniform Taylor-remainder bound that remains valid for the optimizing, $\epsilon$-dependent perturbation. When $P_{X|Y}$ is invertible, this yields a closed-form quadratic surrogate, an explicitly scaled binary mechanism that satisfies the original $\alpha$-lift constraint, a computable sufficient range for $\epsilon$, and an $O(\epsilon^{3/2})$ bound on the utility gap. For $\alpha=2$, the solution is exact over an explicit range. We transfer these results to point-wise R\'enyi-divergence privacy, including an exact-feasible mechanism, an explicit privacy-budget range, and matching gap bounds. We also give an exact formulation for arbitrary leakage matrices and evaluate exact-feasible mechanisms numerically.
\end{abstract}


\section{Introduction}
In this paper, as shown in Fig.~\ref{sys1}, an agent observes useful data represented by the random variable (RV) $Y$, which is arbitrarily correlated with private data represented by the RV $X$. Due to privacy constraints, the agent employs a privacy mechanism to generate disclosed data, represented by the RV $U$. The agent’s goal is to design $U$ based on $Y$ so as to reveal as much information as possible about $Y$ while satisfying a given privacy criterion.
 We use Pearson $\chi^2$-information to measure utility and $\alpha$-lift to measure privacy leakage.

  Related works on the statistical privacy mechanism design can be found in
\cite{long,khodam,Khodam22,zarab2,zarab3,seif,duchi,kairouz,feldman2018privacy, Total, asoodeh3, Calmon1,razegh,emma,ang,sep,lopuha,Calmon2}.
 %In this paper, we study the obtained privacy-utility trade-off.
 
 %i.e., for all $x$ and $u$ we require $-\epsilon\leq \log(\tfrac{P_{X|U}(x|u)}{P_{X}(x)})\leq \epsilon$.

\begin{figure}[t]
	\centering
	\includegraphics[width = 0.4\textwidth]{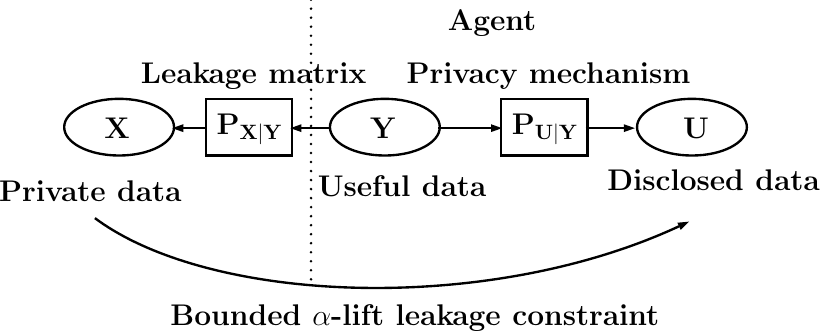}
		\caption{The disclosed data $U$ is designed to maximize $\chi^2(Y;U)$ subject to bounded $\alpha$-lift. The agent has no direct access to $X$; hence, $X-Y-U$ forms a Markov chain.}
	
	\label{sys1}
\end{figure}
\subsection{Related works}
%\textbf{Related works:}
In \cite{sep}, the design of a statistical privacy mechanism that maximizes utility from disclosed data while satisfying privacy constraints is studied. Lower bounds on the privacy--utility trade-off are derived using extended Functional Representation Lemmas and a separation technique.

 %, including scenarios with prioritized private data
 %In \cite{koala}, \emph{secrecy by design} problem is studied under the perfect secrecy assumption. Bounds on secure decomposition have been derived using the Functional Representation Lemma. %and new bounds on privacy-utility trade-off are derived. %The bounds are tight when the private data is a deterministic function of the useful data. 
%In \cite{shah}, the statistical privacy design problems considered in \cite{koala} are generalized by relaxing the perfect secrecy constraint and allowing some leakage. Furthermore, the bounds in \cite{shah} have been tightened in \cite{sep} by using a \emph{separation technique}.
 The concept of maximal leakage is introduced in \cite{issa}, which quantifies the logarithm of the maximum multiplicative increase in an adversary’s probability of correctly guessing any function of a private variable after observing the disclosed data.
The concepts of $\alpha$-leakage and maximal $\alpha$-leakage are introduced in \cite{liao2019tunable}. They are tunable privacy measures that provide different levels of privacy by changing values of $\alpha$. For examples, maximal $\alpha$-leakage equals the Shannon mutual information for $\alpha=1$ and maximal leakage when $\alpha=\infty$.  

In \cite{QIF2025}, earlier $\alpha$-based leakage formulations
are critically examined, and inconsistencies in their
adversarial interpretations and in the distinction between
leakage and leakage capacity are identified. A generalized
quantitative information-flow framework based on
Kolmogorov--Nagumo means is developed to provide consistent
interpretations of these measures. Within this framework,
introducing an auxiliary randomized function of the secret
as the adversary's guessing target is shown to be superfluous:
the corresponding generalized maximal leakage is recovered
as a generalized leakage capacity without this additional
construction. The pointwise information-gain formulation
further identifies posterior-to-prior R\'enyi divergence
as pointwise $\alpha$-leakage and its generalized average
as Sibson mutual information. Related representations and
generalized $g$-leakage interpretations of
$\alpha$-mutual information are developed in
\cite{kamatsuka2025representations,kamatsuka2026generalized}.
These operational interpretations motivate the privacy
criterion considered here, while the present work addresses
mechanism design under an output-wise leakage constraint.

The concept of \textit{lift} is studied
in \cite{zarab2} which represents the likelihood ratio between the posterior and prior beliefs
concerning sensitive features within a dataset.
The concept of \textit{$\alpha$-lift} is introduced in \cite{ding2021alpha}, which is another tunable measure and equals to \textit{max-lift} (worst-case lift) for $\alpha=\infty$. In \cite{zarab3}, a heuristic algorithm is proposed to find the optimal privacy-utility trade-off under bounded $\alpha$-lift leakage constraint. The concept of \textit{multi-level} leakage is studied in \cite{multileakage}, where the goal is to maximize utility under multi-level point-wise privacy constraints, and low-complexity solutions based on local information-geometric approximations are derived. 

As discussed in \cite{khodam} and \cite{Khodam22}, average privacy measures may not be desirable, since certain realizations can leak more information than a prescribed threshold. For instance, average $\chi^2$ and $\ell_1$ measures were used in \cite{Calmon2} and \cite{Total}, respectively, where the average is taken over the letters of $U$. These notions were later strengthened by introducing point-wise (strong) $\chi^2$ and $\ell_1$ privacy measures in \cite{khodam} and \cite{Khodam22}, and it was shown that, using tools from information geometry, the complex design problem can be approximated via linear algebraic techniques.
%The concept of \emph{lift}, which corresponds to another point-wise measure, is studied in \cite{zarab2} and represents the likelihood ratio between the posterior and prior beliefs concerning sensitive features within a dataset.
Concepts from information geometry have been used in several network information theory problems, as well as in privacy, secrecy, and fairness design problems, to approximate complex optimization problems and derive simple designs \cite{Shashi,huang,khodam,Khodam22,razegh,emma,ang,long,multileakage}. Locally, mutual information is one half of a weighted squared Euclidean distance. %to second order. 
Pearson $\chi^2$-information is precisely that quadratic quantity, which motivates using it directly as the utility rather than approximating a mutual-information objective.
Particularly, this approach has been used in privacy context in \cite{khodam,Khodam22,razegh,emma,long,ang,multileakage}, where in \cite{ang}, mutual information is used as the privacy measure and relative entropy as the utility measure, both of which are approximated by quadratic functions. In \cite{long}, the privacy mechanisms are designed considering LIP and max-lift as the privacy leakage measures. %Furthermore, it has been used in \cite{zamani2025fair2,zamani2025fair} to design fair mechanisms under bounded demographic parity and equalized odds constraints.
Furthermore, in \cite{emma}, using information geometry leads to a local approximation of the secrecy capacity over a wire-tap channel.
%Lastly, this approach is used in \cite{multileakage} to address the privacy-utility trade-off under multi-level privacy constraints.
%Lastly, a similar approach has been used in \cite{ang}, where a hypothesis testing problem is considered under a privacy leakage constraint defined by bounded mutual information. In high privacy regimes, the problem is approximated.
\subsection{Contributions}
%\textbf{Contributions:}
In this paper, we study the privacy--utility trade-off under the same $\alpha$-lift constraint as in \cite{zarab3}, but with $\chi^2$-information as utility. This choice is particularly natural here for four reasons. First, it is an $f$-information and therefore obeys the data-processing inequality. Second, it vanishes exactly when $Y$ and $U$ are independent. Third, it equals the prior-normalized variance of the posterior probabilities, so it measures the aggregate gain in estimating the indicator functions of the useful symbols. Finally, when $\max_{u,y}|P_{Y|U}(y|u)/P_Y(y)-1|\to0$, it retains the local behavior of mutual information, $I(Y;U)=\tfrac12\chi^2(Y;U)+o(\chi^2(Y;U))$, while making the quadratic utility expression exact \cite{Calmon2,khodam}.
Our contribution is consequently sharper than a two-sided local approximation: only the $\alpha$-lift constraint is expanded. For invertible $P_{X|Y}$, we solve the resulting quadratic problem in closed form and obtain a binary mechanism from one singular vector. We then control the discarded term uniformly, shrink the surrogate radius to guarantee exact feasibility, derive an explicit sufficient $\epsilon$ range, and sandwich the original optimum between computable bounds whose absolute gap is $O(\epsilon^{3/2})$. 
The remainder bound is established on an $\ell_\infty$-neighborhood of the prior: its auxiliary radius $\rho\in(0,1)$ bounds the normalized posterior perturbation
$
\max_x\left|\frac{P_{X|U}(x|u)}{P_X(x)}-1\right|.
$
We give an admissible choice of $\rho$ computed from $P_{XY}$, $\alpha$, and $\epsilon$.
Using the exact relation between $\alpha$-lift and R\'enyi divergence of the same order, we also obtain an exactly feasible, asymptotically optimal design for the point-wise R\'enyi-private problem and state its mechanism, budget range, and approximation gap explicitly. Finally, we formulate the arbitrary-matrix case without approximating either the utility or the privacy constraint and numerically evaluate the exact-feasible construction.
\section{System Model and Problem Formulation} \label{system}
Let $P_{XY}$ denote the joint distribution of discrete random variables $X$ and $Y$ defined on finite alphabets $\mathcal{X}$ and $\mathcal{Y}$ with equal cardinality, i.e., $|\mathcal{X}|=|\mathcal{Y}|=\mathcal{K}\geq2$. We represent $P_{XY}$ by a matrix defined on $\mathbb{R}^{\mathcal{K}\times\mathcal{K}}$ and %We assume that $X$ and $Y$are defined on spaces that have the same cardinality, i.e., $|\mathcal{X}|=|\mathcal{Y}|=\mathcal{K}\geq2$. 
marginal distributions of $X$ and $Y$ by vectors $P_X$ and $P_Y$ defined on $\mathbb{R}^\mathcal{K}$. We assume that each element in vectors $P_X$ and $P_Y$ is non-zero. Furthermore, we represent the leakage matrix $P_{X|Y}$ by a matrix defined on $\mathbb{R}^{\mathcal{K}\times\mathcal{K}}$ which is assumed to be invertible. Furthermore, for a given $u\in \mathcal{U}$, $P_{X|U}(\cdot|u)$ defined on $\mathbb{R}^{|\mathcal{X}|}$ is a distribution vector with elements $P_{X|U}(x|u)$ for all $x\in\mathcal{X}$. 
All logarithms are natural, $\epsilon\geq0$, and only output symbols of positive probability are retained. The relation between $U$ and $Y$ is described by the kernel $P_{U|Y}$ defined on $\mathbb{R}^{|\mathcal{U}|\times\mathcal{K}}$. 
In this work, $P_{X}(x)$, $P_{X}$, and $\sqrt{P_{X}}$ denote $P(X=x)$, the distribution vector of $X$, and the vector with entries $\sqrt{P_X(x)}$, respectively. We write $D_X\triangleq[P_X]$ and $D_Y\triangleq[P_Y]$ for the corresponding diagonal matrices. For distributions $P$ and $Q$ on the same finite alphabet, with $Q$ strictly positive, the Pearson $\chi^2$-divergence is $\chi^2(P\|Q)=\sum_a(P(a)-Q(a))^2/Q(a)$. We define the $\chi^2$-information between $Y$ and $U$ as
$
\chi^2(Y;U)\triangleq\chi^2(P_{YU}\|P_YP_U)
=\sum_uP_U(u)\chi^2(P_{Y|U=u}\|P_Y).
$
Equivalently, $\chi^2(Y;U)=\sum_y\operatorname{Var}(P_{Y|U}(y|U))/P_Y(y)$. This also has a direct estimation interpretation. If $B_y=\bm1\{Y=y\}$, the law of total variance gives
$
\chi^2(Y;U)=\sum_y\frac{1}{P_Y(y)}\bigl(\mathbb E[(B_y-P_Y(y))^2]
-\mathbb E[(B_y-P_{Y|U}(y|U))^2]\bigr).
$
Thus, the utility is the aggregate prior-normalized reduction in Bayes mean-square error for estimating the useful-symbol indicators from $U$.
To verify the interpretation, write $q_U(y)=P_{Y|U}(y|U)$. Total probability gives $\mathbb E[q_U(y)]=P_Y(y)$, so the defining sum of squared deviations is exactly
$\sum_y\operatorname{Var}(q_U(y))/P_Y(y)$. Also,
$\mathbb E[B_y\mid U]=q_U(y)$, and the variance decomposition for $B_y$ yields the displayed reduction. Nonnegativity is immediate; equality to zero holds precisely when every positive-probability posterior equals $P_Y$, equivalently $Y$ and $U$ are independent.

For clarity, the local mutual-information comparison in the introduction requires uniform posterior closeness. Set
$g_u(y)=q_u(y)/P_Y(y)-1$ and $\delta=\max_{u,y}|g_u(y)|<1$. Taylor expansion of $f(t)=(1+t)\log(1+t)$ uses $f(0)=0$, $f'(0)=f''(0)=1$, and $f'''(t)=-(1+t)^{-2}$. Since $\sum_yP_Y(y)g_u(y)=0$, averaging the expansions gives
\begin{align}
\left|I(Y;U)-\tfrac12\chi^2(Y;U)\right|
\leq \frac{\delta}{6(1-\delta)^2}\chi^2(Y;U).
\label{mi-comparison}
\end{align}
This proves the claimed local relation when $\delta\to0$, even along varying mechanisms. In the invertible model below,
$g_u=D_Y^{-1}P_{X|Y}^{-1}D_Xh_u$, so the uniform bound on $h_u$ in Theorem~\ref{thm:feasible} also ensures $\delta=O(\sqrt\epsilon)$. No such expansion is used to evaluate our utility.
Furthermore, we assume that the private data is not directly accessible by the agent, which implies the Markov chain $X - Y - U$.
Our goal is to design the privacy mechanism that produces $U$, which maximizes the utility and satisfies a privacy criterion.
In this work, utility is measured by $\chi^2(Y;U)$. Privacy leakage is quantified by the bounded $\alpha$-lift given by 
\begin{align}\mathcal{L}^{\alpha}(X \rightarrow u)\triangleq \left(\sum_{x}P_X(x)\left(\frac{P_{X|U}(x|u)}{P_X(x)}\right)^{\alpha}\right)^{\frac{1}{\alpha}}\leq e^{\epsilon},\label{loc}\end{align}
where $\alpha\in(1,\infty)$. We emphasize that \eqref{loc} must hold for each $u\in\mathcal{U}$. However, it corresponds to a semi-pointwise constraint due to the averaging over $X$.
 Thus, the privacy mechanism design problem can be stated as follows 
\begin{subequations}\label{problem}
	\begin{align}
		\sup_{P_{U|Y}} \ \ &\chi^2(Y;U),\label{problem1}\\
	\text{subject to:}\ \ &X-Y-U,\label{Markov1}\\
	 &\mathcal{L}^{\alpha}(X \rightarrow u)\leq e^\epsilon, \forall u\in \mathcal{U},\label{local1}
	\end{align}
\end{subequations}
%where $\mathcal{L}(X \rightarrow u_i)=D_f(P_{X|U=u_i};P_X)$. 
\begin{remark}
\normalfont
We recall that here we focus on the case of $\alpha \in (1,\infty)$. For $\alpha = \infty$, $\mathcal{L}(X \rightarrow u)$ reduces to the max-lift, which has been studied in \cite{long}. More specifically, the privacy–utility trade-off under bounded LIP (Local Information Privacy) and max-lift is studied in \cite{long}, where concepts from information geometry have been applied to approximate the main problem. %Lastly, for $\alpha = 1$, $\mathcal{L}(X \rightarrow U)$ reduces to 
\end{remark}
\begin{remark}
    \normalfont
    The properties of $\alpha$-lift are provided in \cite{zarab3}. For example, $\mathcal{L}^{\alpha}$ is non-decreasing in $\alpha$. More precisely, R\'enyi divergence of order $\alpha$ satisfies
    $D_\alpha(P_{X|U=u}\|P_X)=\frac{\alpha}{\alpha-1}\log\mathcal L^\alpha(X\!\to\!u)$.
    Thus, \eqref{loc} is exactly equivalent to
    $D_\alpha(P_{X|U=u}\|P_X)\leq\alpha\epsilon/(\alpha-1)$, rather than merely implying it. This identity also connects $\alpha$-lift to Sibson mutual information; see \cite{renyi1961measures} and \cite[Remark~1]{zarab3}.
\end{remark}
\begin{remark}
    \normalfont
    As outlined in \cite{zarab3}, increasing $\alpha$ imposes a more stringent privacy criterion; therefore, the optimum in \eqref{problem} is non-increasing in $\alpha$. %In other words, \eqref{problem} is a non-increasing function of $\alpha$.
\end{remark}
Similar to \cite[Proposition~1]{khodam} and \cite[Proposition~2]{Khodam22}, and using similar proof techniques, we obtain the following result.
\begin{proposition}\label{car1}
	It suffices to consider $U$ such that $|\mathcal{U}|\leq|\mathcal{Y}|$. Furthermore, a maximum can be used in \eqref{problem}. %since the
		%corresponding supremum is achieved.
\end{proposition}
\begin{proof}
Put $A\triangleq P_{X|Y}$ and let $\Delta_{\mathcal Y}$ denote the probability simplex on $\mathcal Y$.  A release mechanism can be described by its output probabilities $p_u\triangleq P_U(u)$ and its posterior vectors $q_u\triangleq P_{Y|U=u}$.  Bayes' rule and the law of total probability give
\begin{align}
q_u&\in\Delta_{\mathcal Y},\notag\\
\sum_up_uq_u&=P_Y,\notag\\
P_{X|U=u}&=Aq_u.                                           \label{posterior-mixture}
\end{align}
Conversely, because $P_Y$ has full support, any finite collection satisfying \eqref{posterior-mixture} induces the channel
\begin{align}
P_{U|Y}(u|y)=\frac{p_uq_u(y)}{P_Y(y)}.                     \label{bayes-channel}
\end{align}
The entries in \eqref{bayes-channel} are nonnegative, and the barycenter identity in \eqref{posterior-mixture} shows that they sum to one for every $y$.

Define the exact feasible-posterior set
\begin{align*}
\mathcal Q_\epsilon
\triangleq\left\{q\in\Delta_{\mathcal Y}:
\sum_xP_X(x)\left(\frac{(Aq)(x)}{P_X(x)}\right)^\alpha
\leq e^{\alpha\epsilon}\right\}.
\end{align*}
It is a closed subset of the compact simplex and is therefore compact.  The objective contribution
\begin{align*}
g(q)\triangleq\sum_y\frac{(q(y)-P_Y(y))^2}{P_Y(y)}
\end{align*}
is continuous. To prove attainment and the cardinality bound without assuming finite support, write $n=|\mathcal Y|$ and consider the compact graph
\begin{align*}
\mathcal S=\{(q(1),\ldots,q(n-1),g(q)):q\in\mathcal Q_\epsilon\}
\subset\mathbb R^n.
\end{align*}
Carath\'eodory's theorem says that every point of $\operatorname{conv}(\mathcal S)$ is a convex combination of at most $n+1$ graph points. In particular, this convex hull is compact: it is the continuous image of the compact set of $n+1$ graph points and their simplex of mixing weights. Any posterior distribution, including one with infinitely many outputs, has its vector of averaged coordinates and utility in this closed convex hull. Conversely, every such convex combination with barycenter $P_Y$ reconstructs a finite channel by \eqref{bayes-channel}.

The section of $\operatorname{conv}(\mathcal S)$ whose first $n-1$ coordinates equal those of $P_Y$ is nonempty, since $P_Y\in\mathcal Q_\epsilon$, and compact. Its last coordinate therefore has a maximum. Carath\'eodory's theorem gives an attaining channel with at most $n+1$ outputs.

We can reduce this support to $n$. For any finite optimum with $m>n$ positive weights $p_i$, the vectors $(1,q_i(1),\ldots,q_i(n-1))$ are linearly dependent. Hence some nonzero coefficients $a_i$ satisfy
\begin{align*}
\sum_i a_i=0,\qquad \sum_i a_iq_i=0.
\end{align*}
Both $p_i+ta_i$ and $p_i-ta_i$ remain nonnegative for small $t>0$ and preserve the barycenter. Optimality then forces $\sum_i a_i g(q_i)=0$; otherwise one sign would increase the utility. Since the $a_i$ have both signs, increasing $t$ in $p_i+ta_i$ up to
$\min_{i:a_i<0}p_i/(-a_i)$ removes at least one positive weight while preserving utility and feasibility. Repeating yields an optimum with at most $n$ outputs. This argument uses only full support of $P_Y$ and compactness, so it also applies to arbitrary leakage matrices.
\end{proof}
Finally, in this paper, we also discuss how to generalize our framework to arbitrary alphabet sizes $|\mathcal{X}|$ and $|\mathcal{Y}|$.% as well as to any leakage matrix.
%\vspace{-3mm}
\section{Main Results}
We work with the actual distributional perturbation, rather than a perturbation vector that is held fixed as $\epsilon$ changes. Define
$Z_u\triangleq D_X^{-1/2}(P_{X|U=u}-P_X)$ and the likelihood-ratio perturbation $h_u\triangleq D_X^{-1/2}Z_u$. Let $r_u\triangleq\|Z_u\|_2$, so $r_u^2=\chi^2(P_{X|U=u}\|P_X)$. Normalization and total probability give
\begin{align}
Z_u&\perp\sqrt{P_X},\quad \forall u,\label{proper1}\\
\sum_uP_U(u)Z_u&=0.\label{proper2}
\end{align}
To see these identities explicitly, note that
\begin{align}
P_{X|U=u}(x)=P_X(x)\bigl(1+h_u(x)\bigr),                 \label{likelihood-coordinate}
\end{align}
and therefore
\begin{align}
\sum_xP_X(x)h_u(x)=0,\qquad
r_u^2=\sum_xP_X(x)h_u(x)^2.                              \label{weighted-moments}
\end{align}
The first equality in \eqref{weighted-moments} is exactly $Z_u\perp\sqrt{P_X}$.  Averaging \eqref{likelihood-coordinate} over $U$ yields \eqref{proper2}.

The following bound is uniform over all perturbations in a prescribed neighborhood, including perturbations that depend on $\epsilon$.
\begin{lemma}[Uniform remainder]\label{lem:uniform}
Fix $\rho\in(0,1)$ and suppose $\|h_u\|_\infty\leq\rho$. Define
\begin{align*}
A_\alpha&\triangleq\frac{\alpha(\alpha-1)}2,\
M_{\alpha,\rho}\triangleq
\max\{(1-\rho)^{\alpha-3},(1+\rho)^{\alpha-3}\},\\
K_{\alpha,\rho}&\triangleq
\frac{|\alpha(\alpha-1)(\alpha-2)|}{6}M_{\alpha,\rho}.
\end{align*}
Then, we have
\begin{align}
\bigl(\mathcal L^\alpha(X\!\to\!u)\bigr)^\alpha
&=1+A_\alpha r_u^2+R_u,\nonumber\\
|R_u|&\leq K_{\alpha,\rho}\|h_u\|_\infty r_u^2.\label{uniform}
\end{align}
\end{lemma}
\begin{proof}
Fix $u$ and abbreviate $h=h_u$.  By \eqref{likelihood-coordinate},
\begin{align}
\bigl(\mathcal L^\alpha(X\!\to\!u)\bigr)^\alpha
=\sum_xP_X(x)(1+h(x))^\alpha.                            \label{lift-h}
\end{align}
Because $\|h\|_\infty\leq\rho<1$, all arguments $1+h(x)$ are strictly positive.  For $f(t)=(1+t)^\alpha$, Taylor's theorem at $t=0$ gives, for every $x$,
\begin{align}
f(h(x))
=1+\alpha h(x)+\frac{\alpha(\alpha-1)}2h(x)^2
+\frac{f^{(3)}(\xi_x)}6h(x)^3,                           \label{pointwise-taylor}
\end{align}
where $\xi_x$ lies between $0$ and $h(x)$.  The third derivative is
\begin{align*}
f^{(3)}(t)=\alpha(\alpha-1)(\alpha-2)(1+t)^{\alpha-3}.
\end{align*}
Since $|\xi_x|\leq\rho$,
\begin{align}
|f^{(3)}(\xi_x)|
&\leq |\alpha(\alpha-1)(\alpha-2)|
\max_{t\in[-\rho,\rho]}(1+t)^{\alpha-3}\notag\\
&=6K_{\alpha,\rho}.                                      \label{third-derivative-bound}
\end{align}
Substituting \eqref{pointwise-taylor} in \eqref{lift-h}, the constant terms sum to one and the linear terms vanish by \eqref{weighted-moments}.  The quadratic terms equal
$A_\alpha\sum_xP_X(x)h(x)^2=A_\alpha r_u^2$.  If $R_u$ denotes the sum of the third-order remainders, then \eqref{third-derivative-bound} yields
\begin{align*}
|R_u|
&\leq K_{\alpha,\rho}\sum_xP_X(x)|h(x)|^3\\
&\leq K_{\alpha,\rho}\|h\|_\infty
      \sum_xP_X(x)h(x)^2\\
&=K_{\alpha,\rho}\|h_u\|_\infty r_u^2.
\end{align*}
This proves \eqref{uniform}.  Importantly, the constants depend only on $\alpha$ and the fixed neighborhood radius $\rho$; they do not depend on $u$, $\epsilon$, or the direction of the perturbation.  The estimate is therefore uniform over any optimizing sequence $h_u=h_u(\epsilon)$ that stays inside the prescribed neighborhood.  For $\alpha=2$, $f^{(3)}\equiv0$, so $K_{2,\rho}=0$ and the expansion is exact.
\end{proof}

Let
%\begin{align}
$b_\epsilon\triangleq e^{\alpha\epsilon}-1,\
r_0^2(\epsilon)\triangleq\frac{b_\epsilon}{A_\alpha}.$
%\label{radius0}
%\end{align}
Ignoring $R_u$ in \eqref{uniform} gives the quadratic privacy surrogate $r_u^2\leq r_0^2(\epsilon)$. Let
$W\triangleq D_Y^{-1/2}P_{X|Y}^{-1}D_X^{1/2}$.

\begin{lemma}[Posterior-coordinate identities]\label{lem:coordinates}
For every $u$,
\begin{align}
P_{Y|U=u}-P_Y
&=P_{X|Y}^{-1}D_X^{1/2}Z_u,                              \label{posterior-inverse}\\
\chi^2(P_{Y|U=u}\|P_Y)&=\|WZ_u\|_2^2.                  \label{utility-coordinate}
\end{align}
Moreover,
\begin{align}
W\sqrt{P_X}=\sqrt{P_Y},\qquad
W^T\sqrt{P_Y}=\sqrt{P_X},                               \label{prior-singular-pair}
\end{align}
and $W$ maps the subspace orthogonal to $\sqrt{P_X}$ into the subspace orthogonal to $\sqrt{P_Y}$.  A right singular vector associated with $\sigma_{\max}(W)$ can therefore be chosen orthogonal to $\sqrt{P_X}$.
\end{lemma}
\begin{proof}
The Markov chain gives $P_{X|U=u}=P_{X|Y}P_{Y|U=u}$, while $P_X=P_{X|Y}P_Y$.  Subtracting these identities and using invertibility gives
\begin{align*}
P_{Y|U=u}-P_Y
&=P_{X|Y}^{-1}(P_{X|U=u}-P_X)\\
&=P_{X|Y}^{-1}D_X^{1/2}Z_u,
\end{align*}
which is \eqref{posterior-inverse}.  Consequently,
\begin{align*}
\chi^2(P_{Y|U=u}\|P_Y)
&=(P_{Y|U=u}-P_Y)^TD_Y^{-1}(P_{Y|U=u}-P_Y)\\
&=\|D_Y^{-1/2}P_{X|Y}^{-1}D_X^{1/2}Z_u\|_2^2
=\|WZ_u\|_2^2.
\end{align*}

For the first identity in \eqref{prior-singular-pair}, use
$D_X^{1/2}\sqrt{P_X}=P_X$, $P_{X|Y}^{-1}P_X=P_Y$, and
$D_Y^{-1/2}P_Y=\sqrt{P_Y}$.  For the second, column stochasticity gives
$P_{X|Y}^T\bm1=\bm1$, hence $P_{X|Y}^{-T}\bm1=\bm1$, and therefore
\begin{align*}
W^T\sqrt{P_Y}
=D_X^{1/2}P_{X|Y}^{-T}D_Y^{-1/2}\sqrt{P_Y}
=D_X^{1/2}\bm1=\sqrt{P_X}.
\end{align*}
Thus, if $z\perp\sqrt{P_X}$, then
$\langle Wz,\sqrt{P_Y}\rangle
=\langle z,W^T\sqrt{P_Y}\rangle=0$.

Finally, let $B\triangleq D_X^{-1/2}P_{X|Y}D_Y^{1/2}=W^{-1}$.  For every $a\perp\sqrt{P_Y}$, the vector $q=P_Y+D_Y^{1/2}ta$ is a probability distribution for sufficiently small $|t|$.  For completeness, put $a_y=a(y)$ and $A=P_{X|Y}$. Since $P_X(x)=\sum_y A_{xy}P_Y(y)$, Cauchy--Schwarz gives
\begin{align*}
\|Ba\|_2^2
&=\sum_x\frac{\bigl(\sum_y A_{xy}\sqrt{P_Y(y)}a_y\bigr)^2}{P_X(x)}\\
&\leq\sum_{x,y}A_{xy}a_y^2=\sum_y a_y^2,
\end{align*}
where the last identity uses column stochasticity. Thus $\|Ba\|_2\leq\|a\|_2$ directly.  Hence every singular value of $B$ on the centered subspace is at most one, and every singular value of $W$ on the corresponding centered subspace is at least one.  Equation \eqref{prior-singular-pair} supplies a singular value equal to one in the prior direction.  Therefore the largest singular value of $W$ is attained on the centered subspace (with a centered maximizer also available in the tied case).
\end{proof}

Using Lemma~\ref{lem:coordinates}, the surrogate problem is
\begin{align}\label{tt}
\max_{\substack{P_U,\{Z_u\}:\ \sum_uP_U(u)Z_u=0,\\
\|Z_u\|_2^2\leq r_0^2(\epsilon),\ Z_u\perp\sqrt{P_X}}}
\sum_uP_U(u)\|WZ_u\|_2^2.
\end{align}
The quadratic objective in \eqref{tt} equals $\chi^2(Y;U)$ whenever the induced posteriors are valid. The surrogate omits posterior nonnegativity; Theorem~\ref{thm:feasible} establishes validity for the scaled optimizer.

Let $v^\star$ be a unit dominant right singular vector of $W$ chosen, as permitted by Lemma~\ref{lem:coordinates}, so that $v^\star\perp\sqrt{P_X}$.
\begin{proposition}\label{prop:surrogate}
A binary, equiprobable release is optimal for \eqref{tt}, and its value is
\begin{align}
V_0(\epsilon,\alpha)
=r_0^2(\epsilon)\sigma_{\max}^2(W)
=\frac{2(e^{\alpha\epsilon}-1)}{\alpha(\alpha-1)}\sigma_{\max}^2(W).\label{surrogatevalue}
\end{align}
It is attained in the surrogate problem by $Z_{u_1}=-Z_{u_2}=r_0v^\star$.
\end{proposition}
\begin{proof}
For any feasible $\{P_U(u),Z_u\}$,
\begin{align}
\sum_uP_U(u)\|WZ_u\|_2^2
&\leq\sigma_{\max}^2(W)\sum_uP_U(u)\|Z_u\|_2^2\notag\\
&\leq\sigma_{\max}^2(W)\sum_uP_U(u)r_0^2(\epsilon)\notag\\
&=r_0^2(\epsilon)\sigma_{\max}^2(W).                    \label{surrogate-upper}
\end{align}
Now take two output symbols with probabilities $1/2$ and set
$Z_{u_1}=r_0v^\star$ and $Z_{u_2}=-r_0v^\star$.  The centeredness of $v^\star$ gives $Z_{u_i}\perp\sqrt{P_X}$, both norm constraints are active, and
\begin{align*}
\tfrac12Z_{u_1}+\tfrac12Z_{u_2}=0.
\end{align*}
Thus all constraints of \eqref{tt} hold.  Since
$\|Wv^\star\|_2=\sigma_{\max}(W)$, its objective value is
\begin{align*}
\tfrac12r_0^2\|Wv^\star\|_2^2
+\tfrac12r_0^2\|Wv^\star\|_2^2
=r_0^2\sigma_{\max}^2(W),
\end{align*}
which meets \eqref{surrogate-upper}.  Substituting
$r_0^2=(e^{\alpha\epsilon}-1)/A_\alpha$ and
$A_\alpha=\alpha(\alpha-1)/2$ gives \eqref{surrogatevalue}.
\end{proof}

We next give an explicit exactly feasible mechanism and bound its gap to the original problem. Let $p_{\min}\triangleq\min_xP_X(x)$ and define
\begin{align*}
\kappa_\star&\triangleq\|D_X^{-1/2}v^\star\|_\infty,\
d^\star\triangleq P_{X|Y}^{-1}D_X^{1/2}v^\star,\\
r_Y^\star&\triangleq\min_{y:d^\star(y)\neq0}
\frac{P_Y(y)}{|d^\star(y)|},\
r_{\rm val}^\star\triangleq\min\{\kappa_\star^{-1},r_Y^\star\}.
\end{align*}
Further, set $\delta_0(\epsilon)\triangleq\kappa_\star r_0(\epsilon)$ and
\[
\bar\delta_\epsilon\triangleq
\begin{cases}
\sqrt{\bigl(e^{\alpha\epsilon/(\alpha-1)}-1\bigr)/p_{\min}},&\alpha\geq2,\\
p_{\min}^{-1}\sqrt{2\alpha\epsilon/(\alpha-1)},&1<\alpha<2.
\end{cases}
\]
For $K_{\alpha,\rho}>0$, let
$c_\rho\triangleq\min\{\rho,A_\alpha/(2K_{\alpha,\rho})\}$; when $K_{\alpha,\rho}=0$, let $c_\rho=\rho$. With
$m_\rho\triangleq\min\{\rho/\kappa_\star,r_Y^\star\}$, set
\begin{align}
\epsilon_{\rm dir}&\triangleq
\frac1\alpha\log(1+A_\alpha m_\rho^2),\nonumber\\
\epsilon_{\rm rem}&\triangleq
\begin{cases}
\dfrac{\alpha-1}{\alpha}\log(1+p_{\min}c_\rho^2),&\alpha\geq2,\\[1mm]
\dfrac{(\alpha-1)p_{\min}^2c_\rho^2}{2\alpha},&1<\alpha<2,
\end{cases}\nonumber\\[-1mm]
\epsilon_{\max}(\rho)&\triangleq
\min\{\epsilon_{\rm dir},\epsilon_{\rm rem}\}.
\label{epsilonmax}
\end{align}

\begin{theorem}[Exact feasibility and gap]\label{thm:feasible}
Fix an auxiliary radius $\rho\in(0,1)$. Suppose $0<\epsilon\leq\epsilon_{\max}(\rho)$ and define
\begin{align}
r_-^2(\epsilon;\rho)\triangleq
\frac{b_\epsilon}{A_\alpha+K_{\alpha,\rho}\delta_0(\epsilon)}.\label{saferadius}
\end{align}
The binary mechanism with $P_U(u_1)=P_U(u_2)=1/2$ and
\begin{align}
P_{X|U=u_{1,2}}&=P_X\pm r_-D_X^{1/2}v^\star,\label{safex}\\
P_{Y|U=u_{1,2}}&=P_Y\pm r_-d^\star\label{safey}
\end{align}
is a valid channel and satisfies the original $\alpha$-lift constraint exactly. If $V_{\chi^2}(\epsilon,\alpha)$ denotes the optimum of \eqref{problem}, then
\begin{align}
\frac{b_\epsilon\sigma_{\max}^2(W)}{A_\alpha+K_{\alpha,\rho}\delta_0}
\leq V_{\chi^2}(\epsilon,\alpha)
\leq\frac{b_\epsilon\sigma_{\max}^2(W)}{A_\alpha-K_{\alpha,\rho}\bar\delta_\epsilon}.
\label{sandwich}
\end{align}
Consequently, with
\begin{align*}
G_-\triangleq
\frac{K_{\alpha,\rho}\delta_0}
{A_\alpha(A_\alpha+K_{\alpha,\rho}\delta_0)},\
G_+\triangleq
\frac{K_{\alpha,\rho}\bar\delta_\epsilon}
{A_\alpha(A_\alpha-K_{\alpha,\rho}\bar\delta_\epsilon)},
\end{align*}
we have
\begin{align}
|V_{\chi^2}-V_0|\leq b_\epsilon\sigma_{\max}^2(W)
\max\{G_-,G_+\}.\label{gap}
\end{align}
For fixed $P_{XY},\alpha,\rho$, the absolute gap is $O(\epsilon^{3/2})$ and $|V_{\chi^2}-V_0|/V_0=O(\sqrt\epsilon)$ as $\epsilon\downarrow0$.
\end{theorem}
\begin{proof}
We divide the argument into five steps.

\emph{1) Consequences of the explicit budget range.}
The inequality $\epsilon\leq\epsilon_{\rm dir}$ gives
\begin{align*}
b_\epsilon=e^{\alpha\epsilon}-1\leq A_\alpha m_\rho^2,
\end{align*}
and hence
\begin{align}
r_0(\epsilon)&\leq m_\rho
\leq\frac{\rho}{\kappa_\star},\notag\\
r_0(\epsilon)&\leq r_Y^\star.                             \label{direction-range-proof}
\end{align}
Thus $\delta_0=\kappa_\star r_0\leq\rho$.  Next consider the remainder range.  If $\alpha\geq2$, then
\begin{align*}
\epsilon\leq\epsilon_{\rm rem}
\quad\Longrightarrow\quad
e^{\alpha\epsilon/(\alpha-1)}-1\leq p_{\min}c_\rho^2,
\end{align*}
so $\bar\delta_\epsilon\leq c_\rho$.  If $1<\alpha<2$, then
\begin{align*}
\epsilon\leq\epsilon_{\rm rem}
\quad\Longrightarrow\quad
\frac{1}{p_{\min}}\sqrt{\frac{2\alpha\epsilon}{\alpha-1}}
\leq c_\rho,
\end{align*}
which gives the same conclusion.  Therefore, in both cases,
\begin{align}
\bar\delta_\epsilon&\leq c_\rho\leq\rho,\notag\\
K_{\alpha,\rho}\bar\delta_\epsilon&\leq\frac{A_\alpha}{2},             \label{remainder-range-proof}
\end{align}
where the last inequality is automatic if $K_{\alpha,\rho}=0$.  In particular,
$A_\alpha-K_{\alpha,\rho}\bar\delta_\epsilon\geq A_\alpha/2>0$.

\emph{2) Validity of the proposed posteriors.}
Since the denominator in \eqref{saferadius} is at least $A_\alpha$, we have $r_-\leq r_0$.  For the two signs in \eqref{safex},
\begin{align*}
P_{X|U=u_{1,2}}(x)
=P_X(x)\left(1\pm r_-\frac{v^\star(x)}{\sqrt{P_X(x)}}\right).
\end{align*}
The absolute value of the perturbation in parentheses is at most
$r_-\kappa_\star\leq\delta_0\leq\rho<1$, so these vectors are nonnegative.  They sum to one because
$\langle v^\star,\sqrt{P_X}\rangle=0$.  Likewise,
$r_-\leq r_0\leq r_Y^\star$ ensures that the two vectors in \eqref{safey} are nonnegative.  They are normalized because
\begin{align*}
\bm1^Td^\star
=\bm1^TP_{X|Y}^{-1}D_X^{1/2}v^\star
=\bm1^TD_X^{1/2}v^\star
=\langle\sqrt{P_X},v^\star\rangle=0,
\end{align*}
where $\bm1^TP_{X|Y}^{-1}=\bm1^T$ follows from column stochasticity.  Finally,
\begin{align*}
P_{X|Y}(P_Y\pm r_-d^\star)
=P_X\pm r_-D_X^{1/2}v^\star,
\end{align*}
so the $X$- and $Y$-posteriors are consistent with the Markov chain.  The equiprobable average of either pair equals its prior.  Therefore \eqref{bayes-channel} constructs a valid binary channel $P_{U|Y}$ having precisely these posteriors.

\emph{3) Exact feasibility and the achievable lower bound.}
For the proposed $X$-posteriors,
\begin{align*}
h_{u_{1,2}}&=\pm r_-D_X^{-1/2}v^\star,\\
\|h_{u_i}\|_\infty&=r_-\kappa_\star\leq\delta_0\leq\rho.
\end{align*}
Lemma~\ref{lem:uniform} and the upper sign of its remainder estimate give
\begin{align*}
\bigl(\mathcal L^\alpha(X\!\to\!u_i)\bigr)^\alpha
&\leq1+A_\alpha r_-^2
  +K_{\alpha,\rho}\|h_{u_i}\|_\infty r_-^2\\
&\leq1+(A_\alpha+K_{\alpha,\rho}\delta_0)r_-^2\\
&=1+b_\epsilon=e^{\alpha\epsilon}.
\end{align*}
Thus the original, unexpanded privacy constraint holds for both outputs.  By \eqref{utility-coordinate} and the singular-vector property,
\begin{align*}
\chi^2(Y;U)
&=\tfrac12\|Wr_-v^\star\|_2^2
 +\tfrac12\|W(-r_-v^\star)\|_2^2\\
&=r_-^2\sigma_{\max}^2(W),
\end{align*}
which is the lower bound in \eqref{sandwich}.

\emph{4) A uniform outer bound for every exactly feasible mechanism.}
Fix any output $u$ of any mechanism feasible for \eqref{problem}.  The lift--R\'enyi identity gives
\begin{align}
D_\alpha(P_{X|U=u}\|P_X)\leq\frac{\alpha}{\alpha-1}\epsilon.           \label{renyi-from-lift-proof}
\end{align}
If $\alpha\geq2$, monotonicity of R\'enyi divergence in its order \cite[Theorem~3]{vanerven2014} implies
$D_2\leq D_\alpha$.  Moreover,
\begin{align*}
D_2(P_{X|U=u}\|P_X)
&=\log\sum_x\frac{P_{X|U}(x|u)^2}{P_X(x)}\\
&=\log(1+r_u^2).
\end{align*}
Combining this identity with \eqref{renyi-from-lift-proof} yields
\begin{align*}
r_u^2\leq e^{\alpha\epsilon/(\alpha-1)}-1,
\end{align*}
and hence, for every $x$,
\begin{align*}
|h_u(x)|=\frac{|Z_u(x)|}{\sqrt{P_X(x)}}
\leq\frac{r_u}{\sqrt{p_{\min}}}
\leq\bar\delta_\epsilon.
\end{align*}
If $1<\alpha<2$, monotonicity instead gives
$D(P_{X|U=u}\|P_X)\leq D_\alpha(P_{X|U=u}\|P_X)$, where $D$ is relative entropy.  Pinsker's inequality \cite[Theorem~31]{vanerven2014} (with natural logarithms) and \eqref{renyi-from-lift-proof} then give
\begin{align*}
\|P_{X|U=u}-P_X\|_1
\leq\sqrt{2D(P_{X|U=u}\|P_X)}
\leq\sqrt{\frac{2\alpha\epsilon}{\alpha-1}}.
\end{align*}
Therefore
\begin{align*}
\|h_u\|_\infty
\leq\frac{\|P_{X|U=u}-P_X\|_1}{p_{\min}}
\leq\bar\delta_\epsilon.
\end{align*}
Together with \eqref{remainder-range-proof}, both cases show that
$\|h_u\|_\infty\leq\bar\delta_\epsilon\leq\rho$.  The lower side of Lemma~\ref{lem:uniform} now gives
\begin{align*}
e^{\alpha\epsilon}
&\geq\bigl(\mathcal L^\alpha(X\!\to\!u)\bigr)^\alpha\\
&\geq1+\bigl(A_\alpha-K_{\alpha,\rho}\|h_u\|_\infty\bigr)r_u^2\\
&\geq1+\bigl(A_\alpha-K_{\alpha,\rho}\bar\delta_\epsilon\bigr)r_u^2.
\end{align*}
Thus
\begin{align}
r_u^2\leq
\frac{b_\epsilon}{A_\alpha-K_{\alpha,\rho}\bar\delta_\epsilon}.       \label{exact-radius-outer}
\end{align}
Using \eqref{utility-coordinate}, the singular-value inequality, and $\sum_uP_U(u)=1$, we obtain
\begin{align*}
\chi^2(Y;U)
&=\sum_uP_U(u)\|WZ_u\|_2^2\\
&\leq\sigma_{\max}^2(W)\sum_uP_U(u)r_u^2\\
&\leq\frac{b_\epsilon\sigma_{\max}^2(W)}
{A_\alpha-K_{\alpha,\rho}\bar\delta_\epsilon}.
\end{align*}
Taking the maximum over feasible mechanisms proves the upper bound in \eqref{sandwich}.

\emph{5) Gap and asymptotic orders.}
The surrogate value is
$V_0=b_\epsilon\sigma_{\max}^2(W)/A_\alpha$.  Its difference from the achievable lower bound equals
\begin{align*}
b_\epsilon\sigma_{\max}^2(W)
\frac{K_{\alpha,\rho}\delta_0}
{A_\alpha(A_\alpha+K_{\alpha,\rho}\delta_0)}
=b_\epsilon\sigma_{\max}^2(W)G_-.
\end{align*}
The difference between the outer bound and $V_0$ is similarly
$b_\epsilon\sigma_{\max}^2(W)G_+$.  Since the exact optimum lies between those two bounds, \eqref{gap} follows.

Finally,
$b_\epsilon=\alpha\epsilon+O(\epsilon^2)$,
$r_0=O(\sqrt\epsilon)$, and $\delta_0=O(\sqrt\epsilon)$.  The definitions of $\bar\delta_\epsilon$ in both order regimes also give
$\bar\delta_\epsilon=O(\sqrt\epsilon)$.  Hence $G_-=O(\sqrt\epsilon)$ and $G_+=O(\sqrt\epsilon)$, so the absolute gap is $O(\epsilon^{3/2})$.  Since $V_0=\Theta(\epsilon)$ for $\sigma_{\max}(W)>0$, $|V_{\chi^2}-V_0|/V_0=O(\sqrt\epsilon)$. The lower bound is also $\Theta(\epsilon)$, so normalization by $V_{\chi^2}$ gives the same order. These orders are for fixed $P_{XY},\alpha,\rho$.

The actual suboptimality of the constructed channel should be distinguished from the surrogate-value error. If $V_-=r_-^2\sigma_{\max}^2(W)$, then
\begin{align}
0\leq V_{\chi^2}-V_-
\leq b_\epsilon\sigma_{\max}^2(W)(G_-+G_+)
=O(\epsilon^{3/2}).\label{mechanism-gap}
\end{align}
This follows by subtracting the achievable lower bound from the upper bound in \eqref{sandwich}.
\end{proof}

\subsection{Interpreting and selecting the auxiliary radius}\label{subsec:rho}
The radius $\rho$ is chosen for the proof; it is not an unknown parameter of the data model or the optimal channel. The exact gap
$\Delta(\epsilon,\alpha)=|V_{\chi^2}(\epsilon,\alpha)-V_0(\epsilon,\alpha)|$
is independent of $\rho$. However, its explicit certificate
\begin{align}
E_\rho(\epsilon,\alpha)
&\triangleq b_\epsilon\sigma_{\max}^2(W)
\max\{G_-(\rho),G_+(\rho)\},\label{rho-certificate}
\end{align}
the sufficient range $\epsilon_{\max}(\rho)$, and the constructed scale $r_-(\epsilon;\rho)$ generally depend on $\rho$. The latter dependence means that the actual suboptimality of the analytically scaled channel may also depend on $\rho$.

Fix a dominant centered direction $v^\star$ once for the given $P_{XY}$; when the dominant singular value is repeated, any deterministic choice in that centered eigenspace is allowed. Define
\begin{align*}
\mathcal R_\epsilon
&=\{\rho\in(0,1):\epsilon\leq\epsilon_{\max}(\rho)\},\\
\rho_\epsilon&=\max\{\delta_0(\epsilon),\bar\delta_\epsilon\}.
\end{align*}
Taking an infimum over $\mathcal R_\epsilon$ removes any arbitrary choice from the certificate. In fact, no numerical minimization is necessary.

\begin{proposition}[A computable choice of $\rho$]\label{prop:rho-choice}
For $\epsilon>0$, $\mathcal R_\epsilon$ is nonempty if and only if
\begin{align}
\rho_\epsilon&<1,\qquad r_0(\epsilon)\leq r_Y^\star,\notag\\
K_{\alpha,\rho_\epsilon}\bar\delta_\epsilon&\leq A_\alpha/2. \label{rho-test}
\end{align}
When these conditions hold, $\rho_\epsilon\in\mathcal R_\epsilon$ and
\begin{align}
\Delta(\epsilon,\alpha)
\leq E_{\rho_\epsilon}(\epsilon,\alpha)
=\min_{\rho\in\mathcal R_\epsilon}E_\rho(\epsilon,\alpha).
\label{rho-best}
\end{align}
Thus, after fixing the singular direction, the bound and the analytic mechanism can be evaluated from $P_{XY},\alpha,\epsilon$ alone.
\end{proposition}
\begin{proof}
Inverting the logarithms and square roots in \eqref{epsilonmax} shows that membership in $\mathcal R_\epsilon$ is equivalent to
\begin{align}
\delta_0&\leq\rho<1,\qquad r_0\leq r_Y^\star,\notag\\
\bar\delta_\epsilon&\leq\rho,\qquad
K_{\alpha,\rho}\bar\delta_\epsilon\leq A_\alpha/2.
\label{rho-equivalent}
\end{align}
The last inequality is automatic for $K_{\alpha,\rho}=0$. Consequently, every admissible radius is at least $\rho_\epsilon$.

The constant $K_{\alpha,\rho}$ is non-decreasing in $\rho$. Indeed, for $1<\alpha<3$ the maximum defining $M_{\alpha,\rho}$ equals $(1-\rho)^{\alpha-3}$, which is non-decreasing; for $\alpha>3$ it equals $(1+\rho)^{\alpha-3}$, which is also non-decreasing. For $\alpha=2$ the prefactor vanishes, and for $\alpha=3$ the constant is one. If any $\rho$ satisfies \eqref{rho-equivalent}, replacing it by $\rho_\epsilon$ therefore preserves every inequality. This proves necessity in \eqref{rho-test}; sufficiency follows by direct substitution.

For fixed $\epsilon$, all dependence of $E_\rho$ on $\rho$ occurs through $K=K_{\alpha,\rho}$. On the admissible range the denominators are positive, and direct differentiation gives
\begin{align*}
\frac{d}{dK}\frac{K\delta_0}{A_\alpha(A_\alpha+K\delta_0)}
&=\frac{\delta_0}{(A_\alpha+K\delta_0)^2}\geq0,\\
\frac{d}{dK}\frac{K\bar\delta_\epsilon}{A_\alpha(A_\alpha-K\bar\delta_\epsilon)}
&=\frac{\bar\delta_\epsilon}{(A_\alpha-K\bar\delta_\epsilon)^2}\geq0.
\end{align*}
Hence $E_\rho$ is non-decreasing on $\mathcal R_\epsilon$, so the smallest admissible radius minimizes it. Applying Theorem~\ref{thm:feasible} at $\rho_\epsilon$ proves \eqref{rho-best}.
\end{proof}

These conditions are sufficient for the theorem's certificate; their failure does not imply that the original privacy problem is infeasible. A constant release is always feasible. Also, smaller $\rho$ does not automatically enlarge or shrink $\epsilon_{\max}(\rho)$ monotonically, since both neighborhood coverage and the remainder constant enter that minimum.

Although $\rho_\epsilon$ depends on $\epsilon$, the remainder remains uniform: $\rho_\epsilon=O(\sqrt\epsilon)$, and eventually $\rho_\epsilon\leq1/2$, so $K_{\alpha,\rho_\epsilon}\leq K_{\alpha,1/2}$. Therefore the same $O(\epsilon^{3/2})$ absolute and $O(\sqrt\epsilon)$ relative orders hold for this data-dependent choice. %For example, in the medical model presented in Example \ref{ex:medical} Sec. \ref{subsec:numerical}, at $\alpha=3$ and $\epsilon=0.02$, $\delta_0=0.2845$, $\bar\delta_\epsilon=0.3902$, and $\rho_\epsilon=0.3902$ is admissible.
For example, for the medical model presented in Example \ref{ex:medical} in Sec.  \ref{subsec:numerical}, when $\alpha=3$ and $\epsilon=0.02$, we have $\delta_0=0.2845$ and $\bar{\delta}\epsilon=0.3902$, and $\rho_\epsilon=0.3902$ is admissible.

Two special cases clarify the dependence. For $\alpha=2$, $K_{2,\rho}=0$ and the gap vanishes on the posterior-validity range of Corollary~\ref{cor:alpha2}. For $\alpha=3$, $K_{3,\rho}=1$ for every $\rho$, so the finite-budget certificate and analytic radius are independent of $\rho$ once admissibility holds. The sufficient range can still depend on the selected radius.

\begin{corollary}\label{cor:alpha2}
For $\alpha=2$, $K_{2,\rho}=0$ and the quadratic privacy relation is exact. Whenever
$\epsilon\leq\tfrac12\log(1+(r_{\rm val}^\star)^2)$,
$
V_{\chi^2}(\epsilon,2)=(e^{2\epsilon}-1)\sigma_{\max}^2(W),
$
attained by \eqref{safex}--\eqref{safey} with $r_-^2=e^{2\epsilon}-1$.
\end{corollary}
\begin{proof}
For $\alpha=2$, $A_2=1$ and the third derivative of $(1+t)^2$ is zero.  Equations \eqref{lift-h} and \eqref{weighted-moments} therefore give the exact identity
\begin{align}
\bigl(\mathcal L^2(X\!\to\!u)\bigr)^2=1+r_u^2.           \label{alpha2-exact}
\end{align}
Thus the original privacy constraint is precisely the Euclidean ball
$r_u^2\leq e^{2\epsilon}-1$; there is no Taylor remainder.  The assumed range is equivalent to
$r_0=\sqrt{e^{2\epsilon}-1}\leq r_{\rm val}^\star$, so the two maximizing posterior directions from Proposition~\ref{prop:surrogate} are valid probability distributions in both the $X$ and $Y$ alphabets.  Proposition~\ref{prop:surrogate} now applies to the original problem itself and yields the stated value and mechanism.
\end{proof}

For implementation beyond the conservative analytic range \eqref{epsilonmax}, exact feasibility can still be guaranteed by a one-dimensional search.  The following result justifies the construction used in the numerical section.

\begin{lemma}[Exact directional line search]\label{lem:line-search}
For $s\in\{-1,+1\}$, define
\begin{align*}
F_s(r)\triangleq
\sum_xP_X(x)\left(1+s r\frac{v^\star(x)}{\sqrt{P_X(x)}}\right)^\alpha.
\end{align*}
The domain of each function is $0\leq r\leq r_{\rm val}^\star$.
Each $F_s$ is continuous and nondecreasing.  Consequently, the set of scales for which both exact lift constraints hold is a closed interval $[0,t_\epsilon]$, where
\begin{align}
t_\epsilon\triangleq
\max\left\{r\in[0,r_{\rm val}^\star]:
\max_{s\in\{-1,+1\}}F_s(r)\leq e^{\alpha\epsilon}\right\}.             \label{exact-line-search}
\end{align}
The corresponding equiprobable binary mechanism is valid and has exact utility
$t_\epsilon^2\sigma_{\max}^2(W)$.
\end{lemma}
\begin{proof}
Set $a_x=v^\star(x)/\sqrt{P_X(x)}$.  On the interior of the validity interval,
\begin{align*}
F_s'(r)&=s\alpha\sum_xP_X(x)a_x(1+sra_x)^{\alpha-1},\\
F_s''(r)&=\alpha(\alpha-1)\sum_xP_X(x)a_x^2
(1+sra_x)^{\alpha-2}\geq0.
\end{align*}
Moreover,
$F_s'(0)=s\alpha\sum_xP_X(x)a_x
=s\alpha\langle v^\star,\sqrt{P_X}\rangle=0$.
Convexity therefore implies $F_s'(r)\geq0$ for $r\geq0$, and continuity extends monotonicity to the endpoints.  Both functions equal one at $r=0$, so the joint sublevel set in \eqref{exact-line-search} is nonempty.  Monotonicity and continuity make it a closed interval, and compactness of $[0,r_{\rm val}^\star]$ guarantees the existence of its largest element.

The definition of $r_{\rm val}^\star$ guarantees nonnegativity of the $X$- and $Y$-posteriors, and their normalization and Markov consistency were established in Step 2 of Theorem~\ref{thm:feasible}.  The constraint in \eqref{exact-line-search} is exactly the unexpanded $\alpha$-lift constraint.  Finally, \eqref{utility-coordinate} and the singular-vector identity give utility
$t_\epsilon^2\sigma_{\max}^2(W)$.
\end{proof}

The line search maximizes utility only in the prescribed symmetric binary family; it is not a global solution of \eqref{problem} for general $\alpha$. Within the range of Theorem~\ref{thm:feasible}, its utility is at least $V_-$, so \eqref{mechanism-gap} also bounds its suboptimality. For any tested scale its exact Taylor remainder is computable without $\rho$:
\begin{align*}
R_s(r)&=F_s(r)-1-A_\alpha r^2,\qquad s\in\{-1,+1\}.
\end{align*}
For $\alpha=3$, setting $m_3=\sum_xv^\star(x)^3/\sqrt{P_X(x)}$ yields the exact polynomial
\begin{align*}
F_s(r)=1+3r^2+s m_3r^3.
\end{align*}
Thus both constraints reduce to $3r^2+|m_3|r^3\leq e^{3\epsilon}-1$ together with posterior validity. Computing this directional remainder does not, by itself, bound the gap to the unrestricted optimum; the outer bound in Theorem~\ref{thm:feasible} supplies that guarantee.

\subsection{Point-wise R\'enyi-divergence privacy}\label{subsec:renyi}
We now replace the lift constraint by a point-wise R\'enyi-divergence constraint. For order $\beta>1$, R\'enyi divergence is defined as \cite{renyi1961measures}
\begin{align*}
D_\beta(P\|Q)\triangleq\frac{1}{\beta-1}
\log\sum_aQ(a)\left(\frac{P(a)}{Q(a)}\right)^\beta,
\end{align*}
and denote by $V_{\rm R}(\eta,\beta)$ the optimum of \eqref{problem} when \eqref{local1} is replaced by
\begin{align}
D_\beta(P_{X|U=u}\|P_X)\leq\eta,\qquad\forall u.
\label{renyiconstraint}
\end{align}
The defining identity
\begin{align}
e^{(\beta-1)D_\beta(P_{X|U=u}\|P_X)}
=\bigl(\mathcal L^\beta(X\!\to\!u)\bigr)^\beta
\label{renyilift}
\end{align}
shows that the two privacy notions have exactly the same feasible posteriors after rescaling the budget. In particular,
\begin{align}
V_{\rm R}(\eta,\beta)
=V_{\chi^2}\!\left(\frac{\beta-1}{\beta}\eta,\beta\right).
\label{renyiequiv}
\end{align}
This equality is exact and does not rely on a high-privacy approximation.
Indeed, both sides of \eqref{renyilift} equal
\begin{align*}
\sum_xP_X(x)\left(\frac{P_{X|U}(x|u)}{P_X(x)}\right)^\beta.
\end{align*}
Therefore
$D_\beta(P_{X|U=u}\|P_X)\leq\eta$ if and only if
$\mathcal L^\beta(X\!\to\!u)\leq
e^{(\beta-1)\eta/\beta}$.  Identifying
$\alpha=\beta$ and $\epsilon=(\beta-1)\eta/\beta$ proves \eqref{renyiequiv} directly, including equality of the feasible sets rather than only equality of their local approximations.

For completeness, we state the resulting design and bounds directly in the R\'enyi budget. Write $A_\beta$ and $K_{\beta,\rho}$ as in Lemma~\ref{lem:uniform}, with $\alpha$ replaced by $\beta$, and define
\begin{align*}
b_\eta^{\rm R}\triangleq e^{(\beta-1)\eta}-1,\
(r_{0,\rm R})^2\triangleq\frac{b_\eta^{\rm R}}{A_\beta},\
\delta_{0,\rm R}\triangleq\kappa_\star r_{0,\rm R}.
\end{align*}
The quadratic surrogate value is
\begin{align*}
V_{0,\rm R}
\triangleq(r_{0,\rm R})^2\sigma_{\max}^2(W)
=\frac{2(e^{(\beta-1)\eta}-1)}
{\beta(\beta-1)}\sigma_{\max}^2(W),
\end{align*}
and
\begin{align*}
\bar\delta_\eta^{\rm R}&\triangleq
\begin{cases}
\sqrt{(e^\eta-1)/p_{\min}},&\beta\geq2,\\
p_{\min}^{-1}\sqrt{2\eta},&1<\beta<2.
\end{cases}
\end{align*}
For $K_{\beta,\rho}>0$, set
$c_{\beta,\rho}\triangleq\min\{\rho,A_\beta/(2K_{\beta,\rho})\}$; otherwise set $c_{\beta,\rho}=\rho$. An explicit sufficient range is
\begin{align}
\eta_{\rm dir}^{\rm R}&\triangleq
\frac{1}{\beta-1}\log(1+A_\beta m_\rho^2),\nonumber\\
\eta_{\rm rem}^{\rm R}&\triangleq
\begin{cases}
\log(1+p_{\min}c_{\beta,\rho}^2),&\beta\geq2,\\
\tfrac12p_{\min}^2c_{\beta,\rho}^2,&1<\beta<2,
\end{cases}\nonumber\\
\eta_{\max}^{\rm R}(\rho)&\triangleq
\min\{\eta_{\rm dir}^{\rm R},\eta_{\rm rem}^{\rm R}\}.
\label{renyimax}
\end{align}

\begin{corollary}[R\'enyi-private design]\label{cor:renyi}
Fix $\rho\in(0,1)$. Suppose $0<\eta\leq\eta_{\max}^{\rm R}(\rho)$ and let
\begin{align}
(r_{-,\rm R})^2\triangleq
\frac{b_\eta^{\rm R}}
{A_\beta+K_{\beta,\rho}\delta_{0,\rm R}}.
\label{renyiradius}
\end{align}
The equiprobable binary mechanism in \eqref{safex}--\eqref{safey}, with $r_-$ replaced by $r_{-,\rm R}$, is a valid channel satisfying \eqref{renyiconstraint} exactly. Moreover,
\begin{align}
\frac{b_\eta^{\rm R}\sigma_{\max}^2(W)}
{A_\beta+K_{\beta,\rho}\delta_{0,\rm R}}
\leq V_{\rm R}(\eta,\beta)
\leq\frac{b_\eta^{\rm R}\sigma_{\max}^2(W)}
{A_\beta-K_{\beta,\rho}\bar\delta_\eta^{\rm R}}.
\label{renyisandwich}
\end{align}
Define
\begin{align*}
G_{-,\rm R}\triangleq
\frac{K_{\beta,\rho}\delta_{0,\rm R}}
{A_\beta(A_\beta+K_{\beta,\rho}\delta_{0,\rm R})},\
G_{+,\rm R}\triangleq
\frac{K_{\beta,\rho}\bar\delta_\eta^{\rm R}}
{A_\beta(A_\beta-K_{\beta,\rho}\bar\delta_\eta^{\rm R})}.
\end{align*}
Then
\begin{align}
|V_{\rm R}-V_{0,\rm R}|
\leq b_\eta^{\rm R}\sigma_{\max}^2(W)
\max\{G_{-,\rm R},G_{+,\rm R}\}.
\label{renyigap}
\end{align}
Thus, the absolute and relative gaps are $O(\eta^{3/2})$ and $O(\sqrt\eta)$, respectively. For $\beta=2$, the relation is quadratic exactly and, whenever
$\eta\leq\log(1+(r_{\rm val}^\star)^2)$,
%\begin{align}
$V_{\rm R}(\eta,2)=(e^\eta-1)\sigma_{\max}^2(W).$
%\label{renyibeta2}
%\end{align}
\end{corollary}
\begin{proof}
Set
\begin{align}
\alpha&=\beta,\notag\\
\epsilon&=\frac{\beta-1}{\beta}\eta.                     \label{renyi-substitution}
\end{align}
By \eqref{renyilift}, this substitution maps the exact R\'enyi constraint \eqref{renyiconstraint} to the exact lift constraint \eqref{local1}; no approximation is involved.  Under \eqref{renyi-substitution},
\begin{align*}
b_\epsilon=e^{\alpha\epsilon}-1
=e^{(\beta-1)\eta}-1=b_\eta^{\rm R},
\end{align*}
so $r_0=r_{0,\rm R}$ and $\delta_0=\delta_{0,\rm R}$.  The two cases in the outer neighborhood bound transform as
\begin{align*}
e^{\alpha\epsilon/(\alpha-1)}-1&=e^\eta-1,
&&\beta\geq2,\\
\frac{2\alpha\epsilon}{\alpha-1}&=2\eta,
&&1<\beta<2,
\end{align*}
which gives exactly $\bar\delta_\eta^{\rm R}$.

Next, the directional budget condition
\begin{align*}
\epsilon\leq\frac1\alpha\log(1+A_\alpha m_\rho^2)
\end{align*}
becomes $\eta\leq\eta_{\rm dir}^{\rm R}$.  Similarly, the two cases defining $\epsilon_{\rm rem}$ become the two cases defining $\eta_{\rm rem}^{\rm R}$.  Hence
$\eta\leq\eta_{\max}^{\rm R}(\rho)$ is precisely the sufficient range needed to invoke Theorem~\ref{thm:feasible} after the substitution.

Applying that theorem gives the radius \eqref{renyiradius}, exact feasibility, and the two bounds in \eqref{renyisandwich}.  Applying its gap calculation gives \eqref{renyigap}.  Because
$b_\eta^{\rm R}=O(\eta)$ and both neighborhood radii are $O(\sqrt\eta)$, the same calculation yields the stated absolute and relative orders.  Finally, when $\beta=2$, Corollary~\ref{cor:alpha2} applies with
$\epsilon=\eta/2$; its range becomes
$\eta\leq\log(1+(r_{\rm val}^\star)^2)$ and its value becomes
$(e^\eta-1)\sigma_{\max}^2(W)$.
\end{proof}

Proposition~\ref{prop:rho-choice} applies verbatim under \eqref{renyi-substitution}: take $\rho_\eta^{\rm R}=\max\{\delta_{0,\rm R},\bar\delta_\eta^{\rm R}\}$, test $\rho_\eta^{\rm R}<1$, $r_{0,\rm R}\leq r_Y^\star$, and $K_{\beta,\rho_\eta^{\rm R}}\bar\delta_\eta^{\rm R}\leq A_\beta/2$. The certificate in \eqref{renyigap} is then minimal within this family of certificates and depends only on $P_{XY},\beta,\eta$ after fixing $v^\star$. The same one-dimensional line search can be used with \eqref{renyiconstraint}. Also, for an arbitrary leakage matrix, Theorem~\ref{baghal5} remains exact after replacing
$\Phi_\alpha(q)\leq e^{\alpha\epsilon}$ by
$\Phi_\beta(q)\leq e^{(\beta-1)\eta}$. Numerically, every fixed-order $\alpha$-lift curve below is simultaneously a R\'enyi-private curve under the horizontal-axis conversion $\eta=\beta\epsilon/(\beta-1)$. For the medical example with $\beta=3$, the proved lift range $\epsilon_{\max}=0.0241$ therefore becomes $\eta_{\max}^{\rm R}=0.0361$.

\subsection{Numerical examples}\label{subsec:numerical}
We now evaluate the singular-vector mechanism under the original, unexpanded $\alpha$-lift constraint.  In both examples, the scale is computed from Lemma~\ref{lem:line-search}, so every reported solid curve is exactly feasible.  The examples are separated to distinguish the validation of the analytic bounds from the effect of varying $\alpha$.

\begin{example}[Medical disclosure]\label{ex:medical}
Let $X\in\{\mathrm{A},\mathrm{B},\mathrm{C}\}$ denote a sensitive diagnosis category (viral, bacterial, or cardiac), and let $Y\in\{\mathrm{Low},\mathrm{Med},\mathrm{High}\}$ denote a triage-risk level computed from vital signs and laboratory measurements.  We use
\begin{align*}
P_X&=[0.2,0.3,0.5]^T,\\
P_{X|Y}&=
\begin{bmatrix}
0.85&0.10&0.15\\
0.10&0.80&0.15\\
0.05&0.10&0.70
\end{bmatrix}.
\end{align*}
The matrix is invertible and
\begin{align*}
P_Y=P_{X|Y}^{-1}P_X=[0.0884,0.2376,0.6740]^T.
\end{align*}
The dominant centered singular direction and its gain are
\begin{align*}
v^\star&=[-0.8862,0.1643,0.4332]^T,\\
\sigma_{\max}^2(W)&=3.8970.
\end{align*}
For every $\epsilon$, we solve the scalar problem \eqref{exact-line-search}, form the two posteriors
\begin{align*}
P_{X|U=u_{1,2}}&=P_X\pm t_\epsilon D_X^{1/2}v^\star,\\
P_{Y|U=u_{1,2}}&=P_Y\pm t_\epsilon d^\star,
\end{align*}
and evaluate the exact utility
$\chi^2(Y;U)=t_\epsilon^2\sigma_{\max}^2(W)$.  With $\alpha=3$ and $\rho=0.9$, the constants are
\begin{align*}
\kappa_\star&=1.9816,\\
r_{\rm val}^\star=r_Y^\star&=0.1580,\\
\epsilon_{\max}&=0.0241.
\end{align*}
Thus the full range displayed in Fig.~\ref{fig:medical-long} lies inside the sufficient range of Theorem~\ref{thm:feasible}.
As an independent numerical check, at $\epsilon=0.02$ the analytic lower and upper bounds are $0.07337$ and $0.09234$, the directional utility is $0.07532$, and $V_0=0.08033$. A finite-posterior linear program using 720 feasible boundary directions yields a feasible utility of $0.08071$, with its barycenter and exact privacy constraints checked directly. This is a lower bound on the unrestricted optimum, not a computation of that optimum. It demonstrates both that the symmetric directional construction need not be globally optimal and that $V_0$ need not be an upper bound. The absolute-error certificate accommodates this behavior.

\begin{figure}[t]
\centering
\includegraphics[width=0.92\columnwidth]{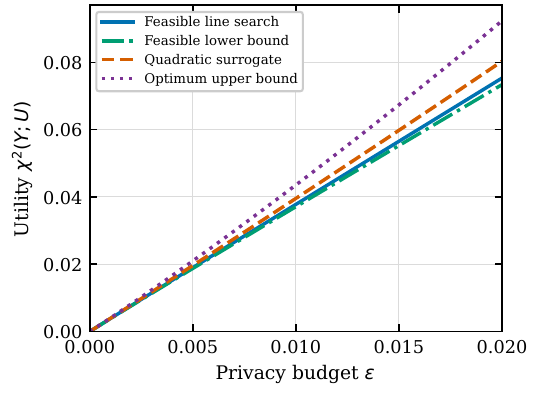}
\caption{Medical disclosure with $\alpha=3$.  The exact directional search is feasible for the original constraint.  The analytic lower and upper bounds are from Theorem~\ref{thm:feasible}, and the quadratic curve is \eqref{surrogatevalue}.}
\label{fig:medical-long}
\end{figure}

The figure illustrates the achievable curve within the proved bounds \eqref{sandwich}; it does not numerically compute the unrestricted optimum.  The analytic radius $r_-$ gives the guaranteed feasible lower curve; the exact directional search can use a larger scale and therefore lies above it.  The upper curve bounds the optimum over all channels, not merely over the chosen direction.  Finally, all curves have the same first-order slope because their pairwise differences are $O(\epsilon^{3/2})$ whereas their leading values are $\Theta(\epsilon)$.
\end{example}

\begin{example}[Smart-meter disclosure]\label{ex:smart-meter}
Let $X\in\{\mathrm{Single},\mathrm{Family},\mathrm{SmallBusiness}\}$ denote a sensitive household type and let $Y\in\{\mathrm{Low},\mathrm{Mid},\mathrm{High}\}$ denote a load-pattern category useful for grid management.  We use
\begin{align*}
P_X&=[0.25,0.35,0.40]^T,\\
P_{X|Y}&=
\begin{bmatrix}
0.80&0.15&0.10\\
0.15&0.70&0.20\\
0.05&0.15&0.70
\end{bmatrix}.
\end{align*}
Here
\begin{align*}
P_Y&=[0.1915,0.3191,0.4894]^T,\\
v^\star&=[-0.3920,0.8058,-0.4438]^T,\\
\sigma_{\max}^2(W)&=3.9410,\\
\kappa_\star&=1.3621,\qquad r_{\rm val}^\star=0.3473.
\end{align*}
For each $\alpha\in\{2,3,5,10\}$ and each privacy budget, we compute $t_\epsilon$ from the two exact constraints in \eqref{exact-line-search}.  The corresponding solid curve in Fig.~\ref{fig:smartmeter-long} is
$t_\epsilon^2\sigma_{\max}^2(W)$, while the dashed curve is the quadratic value
\begin{align*}
V_0(\epsilon,\alpha)=
\frac{2(e^{\alpha\epsilon}-1)}{\alpha(\alpha-1)}
\sigma_{\max}^2(W).
\end{align*}

\begin{figure}[t]
\centering
\includegraphics[width=0.94\columnwidth]{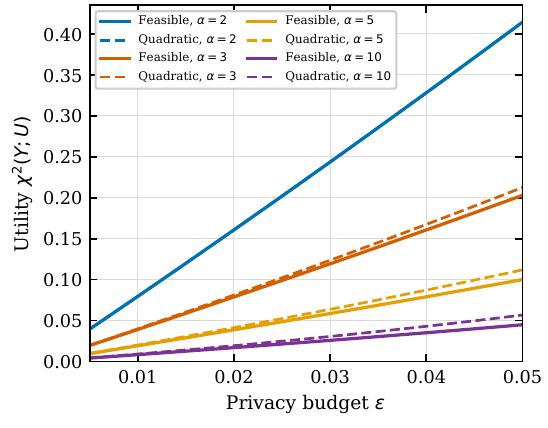}
\caption{Smart-meter disclosure for several orders $\alpha$.  Solid curves use the exact-feasible line search; dashed curves show the quadratic surrogate.  For $\alpha=2$, the two coincide throughout the displayed posterior-validity range.}
\label{fig:smartmeter-long}
\end{figure}

The solid curves illustrate that increasing $\alpha$ tightens the point-wise privacy constraint and cannot increase the best utility in this fixed-direction family.  Corollary~\ref{cor:alpha2} explains the exact coincidence for $\alpha=2$.  For every other order, the solid and dashed curves share the first-order behavior predicted by Theorem~\ref{thm:feasible}; their separation is caused only by the nonlinear privacy constraint because the $\chi^2$ utility itself is exact.
\end{example}

\subsection{General leakage matrices}
We now allow arbitrary alphabet sizes and any column-stochastic leakage matrix $A\triangleq P_{X|Y}$. Exact feasibility can be retained without a Taylor approximation. Define
%\begin{align*}
$\Phi_\alpha(q)\triangleq
\sum_xP_X(x)\left(\frac{(Aq)(x)}{P_X(x)}\right)^\alpha,$ and $
\mathcal Q_\epsilon^{\rm ex}\triangleq
\{q\in\mathbb R_+^{|\mathcal Y|}:\bm1^Tq=1,
\ \Phi_\alpha(q)\leq e^{\alpha\epsilon}\}.$
%\end{align*}
This is a compact convex set for every $\alpha>1$.
Indeed, $q\mapsto Aq$ is affine and $t\mapsto t^\alpha$ is convex on $[0,\infty)$, so $\Phi_\alpha$ is convex.  Its sublevel set is closed and convex; intersecting it with the compact probability simplex gives a compact convex set.  It is nonempty because $AP_Y=P_X$ and therefore $\Phi_\alpha(P_Y)=1\leq e^{\alpha\epsilon}$.

\begin{theorem}\label{baghal5}
For any leakage matrix $A$, the original problem \eqref{problem} is exactly equivalent to
\begin{align}\label{generalproblem}
\max_{P_U,\{q_u\}}\quad
&\sum_uP_U(u)\left\|D_Y^{-1/2}(q_u-P_Y)\right\|_2^2\\[-1mm]
\mathrm{s.t.}\quad &q_u\in\mathcal Q_\epsilon^{\rm ex},\quad
\sum_uP_U(u)q_u=P_Y,\notag\\
&P_U(u)\geq0,\quad\sum_uP_U(u)=1,\quad
|\mathcal U|\leq|\mathcal Y|.\notag
\end{align}
Moreover, an optimizer can be selected with every $q_u$ an extreme point of $\mathcal Q_\epsilon^{\rm ex}$.
\end{theorem}
\begin{proof}
We first prove equivalence in both directions.  Given a channel $P_{U|Y}$, set
$p_u=P_U(u)$ and $q_u=P_{Y|U=u}$.  The Markov chain gives
\begin{align*}
P_{X|U=u}=Aq_u,
\end{align*}
and total probability gives $\sum_up_uq_u=P_Y$.  Furthermore,
\begin{align*}
\bigl(\mathcal L^\alpha(X\!\to\!u)\bigr)^\alpha
&=\sum_xP_X(x)\left(\frac{P_{X|U}(x|u)}{P_X(x)}\right)^\alpha\\
&=\sum_xP_X(x)\left(\frac{(Aq_u)(x)}{P_X(x)}\right)^\alpha
=\Phi_\alpha(q_u).
\end{align*}
Thus the original point-wise constraint is exactly $q_u\in\mathcal Q_\epsilon^{\rm ex}$.  The utility is
\begin{align*}
\chi^2(Y;U)
&=\sum_up_u\sum_y\frac{(q_u(y)-P_Y(y))^2}{P_Y(y)}\\
&=\sum_up_u\left\|D_Y^{-1/2}(q_u-P_Y)\right\|_2^2,
\end{align*}
which is the objective in \eqref{generalproblem}.

Conversely, suppose $\{p_u,q_u\}$ is feasible for \eqref{generalproblem}.  Since $P_Y$ has full support, define
\begin{align*}
P_{U|Y}(u|y)=\frac{p_uq_u(y)}{P_Y(y)}.
\end{align*}
Nonnegativity is immediate, and
\begin{align*}
\sum_uP_{U|Y}(u|y)
=\frac{\sum_up_uq_u(y)}{P_Y(y)}=1.
\end{align*}
Hence this is a valid channel.  Bayes' rule gives posterior $q_u$, the Markov chain produces $Aq_u$, and membership in
$\mathcal Q_\epsilon^{\rm ex}$ gives the original exact privacy constraint.  This proves the equivalence.  Proposition~\ref{car1} supplies existence and the cardinality bound.

It remains to prove the extreme-point claim.  Let
\begin{align*}
g(q)=\left\|D_Y^{-1/2}(q-P_Y)\right\|_2^2.
\end{align*}
This is a convex function of $q$.  If an optimal posterior $q_u$ is not extreme in $\mathcal Q_\epsilon^{\rm ex}$, finite-dimensional convexity gives a representation
\begin{align*}
q_u&=\sum_{j=1}^{m_u}\lambda_{u,j}e_{u,j},\\
e_{u,j}&\in\operatorname{ext}(\mathcal Q_\epsilon^{\rm ex}),\\
\lambda_{u,j}&\geq0,\qquad \sum_j\lambda_{u,j}=1.
\end{align*}
Split output $u$ into symbols $(u,j)$ with probabilities
$p_u\lambda_{u,j}$ and posteriors $e_{u,j}$.  The barycenter is unchanged because
\begin{align*}
\sum_{u,j}p_u\lambda_{u,j}e_{u,j}
=\sum_up_uq_u=P_Y,
\end{align*}
and every new posterior remains exactly private.  Convexity gives
\begin{align*}
g(q_u)\leq\sum_j\lambda_{u,j}g(e_{u,j}),
\end{align*}
so splitting cannot decrease utility.  Since the original mechanism was optimal, the split mechanism is also optimal.

This produces a finite optimal mechanism supported entirely on extreme points, although it may use too many outputs.  Restrict attention to the finite collection of extreme points appearing in that split mechanism and maximize the linear objective over their mixing weights subject to the same barycenter.  The split weights are feasible and globally optimal.  An extreme solution of this finite linear program has at most $|\mathcal Y|$ positive weights by the linear-dependence argument in the proof of Proposition~\ref{car1}.  Selecting such a solution preserves optimality, the barycenter, and support on extreme points while restoring the cardinality bound.
\end{proof}

\begin{proposition}[Null-space decomposition]\label{prop:null-space}
Let $A$ be arbitrary and define
$\Delta_u\triangleq A(q_u-P_Y)$.  Every posterior admits the representation
\begin{align}
q_u&=P_Y+A^\dagger\Delta_u+z_u,\notag\\
z_u&\in\operatorname{null}(A),                            \label{null-decomposition}
\end{align}
subject to normalization and nonnegativity.  Its exact utility contribution is
\begin{align}
\left\|D_Y^{-1/2}(A^\dagger\Delta_u+z_u)\right\|_2^2.    \label{null-utility}
\end{align}
If there exists nonzero $z\in\operatorname{null}(A)$ with
$\bm1^Tz=0$ and $P_Y\pm tz\geq0$ for some $t>0$, then positive utility is achievable at perfect privacy.
\end{proposition}
\begin{proof}
Because $\Delta_u=A(q_u-P_Y)$ belongs to the range of $A$, the Moore--Penrose identity $AA^\dagger\Delta_u=\Delta_u$ holds.  Define
\begin{align*}
z_u\triangleq q_u-P_Y-A^\dagger\Delta_u.
\end{align*}
Then
\begin{align*}
Az_u&=A(q_u-P_Y)-AA^\dagger\Delta_u\\&=\Delta_u-\Delta_u=0,
\end{align*}
so $z_u\in\operatorname{null}(A)$ and \eqref{null-decomposition} follows.  Substituting this expression into
$\chi^2(q_u\|P_Y)=\|D_Y^{-1/2}(q_u-P_Y)\|_2^2$ gives \eqref{null-utility}.

For the final claim, take equiprobable posteriors
$q_1=P_Y+tz$ and $q_2=P_Y-tz$.  They are normalized and nonnegative by assumption, and their average is $P_Y$.  Since $Az=0$,
\begin{align*}
Aq_1&=Aq_2\\&=AP_Y\\&=P_X,
\end{align*}
so the lift equals one for both outputs: privacy is perfect.  Their average utility is
\begin{align*}
\tfrac12\|D_Y^{-1/2}tz\|_2^2
+\tfrac12\|D_Y^{-1/2}(-tz)\|_2^2
=t^2\|D_Y^{-1/2}z\|_2^2>0.
\end{align*}
In the present full-support model, the side conditions hold for every nonzero $z\in\operatorname{null}(A)$: $\bm1^Tz=\bm1^TAz=0$, and any $0<t\leq\min_{y:z(y)\neq0}P_Y(y)/|z(y)|$ ensures nonnegativity. Conversely, at $\epsilon=0$, strict convexity of $t^\alpha$ and $\sum_xP_X(x)(Aq_u)(x)/P_X(x)=1$ imply that lift at most one forces $Aq_u=P_X$. If the utility is positive, some $q_u\neq P_Y$, so $q_u-P_Y$ is a nonzero null vector. Thus positive utility at perfect privacy is possible if and only if $\operatorname{null}(A)\neq\{0\}$.

When $A$ is invertible, its null space is trivial and \eqref{null-decomposition} reduces to the inverse-coordinate representation used in Lemma~\ref{lem:coordinates}.  When the null space is nontrivial, Proposition~\ref{prop:null-space} shows exactly how utility can survive even at $\epsilon=0$.
\end{proof}
%\vspace{-2mm}
\section{Conclusion}
%\vspace{-1mm}
In this paper, we used $\chi^2$-information as utility for privacy mechanism design under bounded $\alpha$-lift. The utility is exactly quadratic. For invertible $P_{X|Y}$, a uniform Taylor bound for the privacy constraint yields a closed-form surrogate, an explicitly scaled mechanism that is feasible for the original constraint, a computable range of $\epsilon$, and an $O(\epsilon^{3/2})$ absolute optimality-gap bound. The case $\alpha=2$ is exact over an explicit range. Through the exact lift--R\'enyi identity, the same conclusions hold for point-wise R\'enyi-divergence privacy after an explicit budget conversion. Numerical results verify the exact constraint rather than only its surrogate. For arbitrary leakage matrices, we gave a fully exact posterior formulation that also preserves useful null-space directions under perfect privacy.

%\clearpage
%\begingroup
\footnotesize
\bibliographystyle{IEEEtran}
\bibliography{references_extended}
%\endgroup

\end{document}